\documentclass{llncs}
\usepackage{epsfig,amsfonts}

\newcounter{nnote_counter}
\newcommand{\ignore}[1]{}

\begin{document}

\title{Approximating $\{0,1,2\}$-Survivable Networks with Minimum Number of Steiner Points}

\author{Nachshon Cohen \and Zeev Nutov}
\institute{The Open University of Israel,\\
\email{nachshonc@gmail.com, nutov@openu.ac.il}}
\def \i {p}
\def \t {\lfloor \lg k \rfloor}
\maketitle


\begin{abstract}
We consider low connectivity variants of the 
{\sf Survivable Network with Minimum Number of Steiner Points} ({\sf SN-MSP}) problem: given 
a finite set $R$ of terminals in a metric space $(M,d)$, 
a subset $B \subseteq R$ of ``unstable'' terminals, and 
connectivity requirements $\{r_{uv}: u,v \in R\}$, 
find a minimum size set $S \subseteq M$ of additional points such that 
the unit-disc graph of $R \cup S$ contains $r_{uv}$ pairwise internally edge-disjoint and 
$(B \cup S)$-disjoint $uv$-paths for all $u,v \in R$.
The case when $r_{uv}=1$ for all $u,v \in R$ is the {\sf Steiner Tree with Minimum Number of Steiner Points}
({\sf ST-MSP}) problem, and the case $r_{uv} \in \{0,1\}$ 
is the {\sf Steiner Forest with Minimum Number of Steiner Points} 
({\sf SF-MSP}) problem.
Let $\Delta$ be the maxi\-mum number of points in a unit ball 
such that the distance between any two of them is larger than $1$.
It is known that $\Delta=5$ in $\mathbb{R}^2$. 
The previous known approximation ratio for {\sf ST-MSP} was  
$\lfloor (\Delta+1)/2 \rfloor+1+\epsilon$ in an arbitrary normed space \cite{NY},
and $2.5+\epsilon$ in the Euclidean space $\mathbb{R}^2$ \cite{cheng2008relay}.
Our approximation ratio for {\sf ST-MSP} is $1+\ln(\Delta-1)+\epsilon$  in an arbitrary normed space,
which in $\mathbb{R}^2$ reduces to $1+\ln 4+\epsilon < 2.3863 +\epsilon$.
For {\sf SN-MSP} with $r_{uv} \in \{0,1,2\}$, we give a simple $\Delta$-approximation algorithm.
In particular, for {\sf SF-MSP}, this improves the previous ratio $2\Delta$.
\end{abstract}

\section{Introduction}

\subsection{Problems considered}

A large research effort is focused on developing algorithms for finding a ``cheap'' 
network that satisfies a certain property.
In wired networks, where connecting any two nodes incurs a cost, many problems can be cast 
as finding a subgraph of minimum cost that satisfies some prescribed connectivity requirements. 
Following previous work on min-cost connectivity problems, 
we use the following generic notion of connectivity.

\begin{definition}
Let $G=(V,E)$ be a graph and let $Q \subseteq V$.
The $Q$-connectivity $\lambda^Q_G(u,v)$ of $u,v$ in $G$ is the maximum number 
of pairwise $(E \cup Q \setminus \{u,v\})$-disjoint $uv$-paths in $G$.
Given connectivity requirements $r=\{r_{uv}:u,v \in R \subseteq V\}$ on a subset $R \subseteq V$ of terminals, 
we denote by $D_r=\{uv:u,v \in R, r_{uv} >0\}$ the set of ``demand edges'' of $r$.
We say that $G$ is $(r,Q)$-connected,
or simply $r$-connected if $Q$ is understood, if $\lambda^Q_G(u,v) \geq r_{uv}$ for all $uv \in D_r$. 
\end{definition}

Note that edge-connectivity is the case $Q=\emptyset$ and node-connectivity is the case $Q=V$.
The members of $E \cup Q$ will be called {\em elements}, hence $\lambda^Q_G(u,v)$ is the maximum number 
of pairwise internally element-disjoint $uv$-paths in $G$. 
Variants of the following classic problem were extensively studied in the literature.

\begin{center} 
\fbox{
\begin{minipage}{0.960\textwidth}
\noindent
{\sf Survivable Network} ({\sf SN}) \\
{\em Instance:} 
A graph $G=(V,E)$ with edge costs, $Q \subseteq V$, and connectivity requirements 
$r=\{r_{uv}:uv \in R \subseteq V\}$. \\
{\em Objective:} Find a minimum-cost $(r,Q)$-connected subgraph $H$ of $G$.
\end{minipage}
}
\end{center}

In practical networks the connectivity requirements are rather small, 
usually $r_{uv} \in \{0,1,2\}$ -- so called {\sf $\{0,1,2\}$-SN}.
Particular cases in this setting are 
{\sf Minimum Spanning Tree} ({\sf MST}) ($r_{uv}=1$ for all $u,v \in V$),
{\sf Steiner Tree}                      ($r_{uv}=1$ for all $u,v \in R$) and 
{\sf Steiner Forest}                    ($r_{uv} \in \{0,1\}$ for all $u,v \in R$),
and {\sf $2$-Connected Subgraph}        ($r_{uv}=2$ for all $u,v \in V$).

In wireless networks, the range and the location of the 
transmitters determines the resulting communication network.
We consider adding a minimum number of transmitters such that the resulting
communication network is $(r,Q)$-connected.
If the range of the transmitters is fixed, our goal is to add a minimum number of transmitters, 
and we get the following type of problems.

\begin{definition}
Let $(M,d)$ be a metric space and let $V \subseteq M$. 
The {\em unit-disk graph} of $V$ has node set $V$ and edge set $\{uv: u,v \in V, d(u,v) \leq  1\}$.
\end{definition}

\begin{center} 
\fbox{
\begin{minipage}{0.960\textwidth}
\noindent
{\sf Survivable Network with Minimum Number of Steiner Points} ({\sf SN-MSP}) \\
{\em Instance:} 
\ A finite set $R \subseteq M$ of {\em terminals} in a metric space $(M,d)$, 
a set $B \subseteq R$ of ``unstable'' terminals, 
connectivity requirements $\{r_{uv}:uv \in R\}$. \\
{\em Objective:} Find a minimum size set $S \subseteq M$ such that the unit-disk graph 
of $R \cup S$ is $(r,Q)$-connected, where $Q=B \cup S$.
\end{minipage}
}
\end{center}

As in previous work, we will allow to place several points at the same location, and 
assume that the maximum distance between terminals is polynomial in the number of terminals.

\subsection{Previous work and our results}

On previous work on high connectivity variants of {\sf SN} problem we refer the reader 
to a survey in \cite{KN-sur} and here only mention some work relevant to this paper. 
The {\sf Steiner Tree} problem was studied extensively, c.f. \cite{Z116,Zln2,RZ,PS,BGRS,GORZ}, 
and the currently best approximation ratio for it is $\ln 4+\epsilon$ \cite{BGRS}.
Let $\tau^*$ denote the optimum value of a standard cut-LP relaxation for {\sf SN}
(see Section~\ref{s:main2}). In \cite{GW} is given a combinatorial primal-dual 
algorithm for {\sf Steiner Forest} that computes a solution of cost at most $2 \tau^*$.
For {\sf $\{0,1,2\}$-SN} a similar results is achieved by the iterative rounding method \cite{FJW};
a combinatorial primal-dual algorithm that computes a solution of cost at most $3\tau^*$ is given in \cite{RW}.


We survey some relevant literature on {\sf SN-MSP} problems.
{\sf ST-MSP} is NP-hard even in $\mathbb{R}^2$, 
and arises in various wireless network design problems, c.f. 
\cite{BDHR,Calinescu,CHEN-DU,cheng2008relay,LKN,KKS,MAN-ZEL,NY} for only a sample of papers in the area,
where it is studied both in $\mathbb{R}^2$ and in general metric spaces.
In the latter case, the approximation ratio is usually expressed in terms of the following parameter.
Let $\Delta$ be the maximum number of ``independent'' points in the unit ball,
such that the distance between any two of them is larger than $1$.
It is known \cite{ROB-SAL} that $\Delta$ equals the maximum degree of a minimum-degree
Minimum Spanning Tree in the normed space.
For Euclidean distances we have $\Delta=5$ in $\mathbb{R}^2$ and $\Delta=11$ in $\mathbb{R}^3$,
and in $\mathbb{R}^\ell$ $\Delta$ is at most the Hadwiger number \cite{ROB-SAL}; 
hence $\Delta \leq 2^{0.401\ell(1+o(1))}$, by \cite{KabV}.
  
In finite metric spaces, {\sf ST-MSP} is equivalent to the variant of the 
{\sf Node Weighted Steiner Tree} problem when all terminals have costs $0$ and the other nodes have cost $1$.
Klein and Ravi \cite{KRa} proved that this variant is {\sf Set-Cover} 
hard to approximate, and gave an $O(\ln |R|)$-approximation algorithm for general weights.
Hence up to constants, even for finite metric spaces,
the ratio $O(\ln |R|)$ of \cite{KRa} is the best possible unless P=NP.
Note however, that this does not exclude constant ratios for metric spaces with small $\Delta$,
e.g., $\Delta=5$ in $\mathbb{R}^2$.

Most algorithms for {\sf SN-MSP} problems applied the following reduction method,
by solving the corresponding {\sf SN} instance obtained as follows.

\begin{definition} \label{d:K}
Given a finite set $R$ of points in a metric space $(M,d)$ and an integer $k \geq 1$,
the (multi)graph $K_R$ has node set $R$ and $k$ parallel edges between every pair of nodes.
The costs of the $k$ edges between $u,v$ are defined as follows.
Let $\hat{d}_{uv}=\max\{\lceil d(u,v) \rceil-1,0\}$.
If $\hat{d}_{uv}>0$, then all the $k$ edges have cost $\hat{d}_{uv}$.
If $\hat{d}_{uv}=0$, then one edge has cost $0$ and the others have cost $1$.
\end{definition}

Let ${\sf opt}$ denote the optimal solution value of a problem instance at hand.
It is easy to see that any solution of cost $C$ to the corresponding {\sf SN} instance
with $k=\max_{uv \in D_r} r_{uv}$
defines a solution $S$ of size $C$ to the original {\sf SN-MSP} instance,
where every node in $S$ has degree exactly $2$; such a solution is called a {\em bead solution}.
Conversely, any bead solution $S$ can be converted into a solution to the {\sf SN} 
instance of cost at most $|S|$ (see \cite{LKN,Calinescu}).
Due to this bijective correspondence, we simply define a bead 
solution as a solution to the corresponding {\sf SN} instance,
and denote the optimal value of a bead solution to an instance $I$ by $\tau=\tau(I)$.
If the {\sf SN} instance admits a $\rho$-approximation algorithm, 
and if for the given {\sf SN-MSP} instance there exists a bead solution $S$ of size $\leq \alpha{\sf opt}$,
then we get a $\rho \alpha$-approximation algorithm for the {\sf SN-MSP} instance.
Equivalently, for a class ${\cal I}$ of {\sf SN-MSP} instances, define a parameter $\alpha$ by 
$\alpha=\alpha({\cal I})=\sup_{I \in {\cal I}} \frac{{\sf opt}(I)}{\tau(I)}$.
Then approximation ratio $\rho$ for {\sf SN} instances that correspond to the class ${\cal I}$ 
implies approximation ratio $\alpha \rho$ for {\sf SN-MSP} instances in class ${\cal I}$.

M\u{a}ndoiu and Zelikovsky \cite{MAN-ZEL} showed that for {\sf ST-MSP} $\alpha=\Delta-1$.
Since the instance of {\sf SN} that corresponds to {\sf ST-MSP} is the {\sf MST} problem 
that can be solved in polynomial time, this gives a $(\Delta-1)$-approximation algorithm for {\sf ST-MSP}.
A more general method, uses a reduction to the 
{\sf Minimum $k$-Connected Spanning Subhypergraph} problem, see Section~\ref{s:main1}.
This method was initiated by Zelikovsky \cite{Z116}, improved in a long series of papers 
(part of them are \cite{Z116,PS,RZ}), 
and culminated in the paper of Byrka, Grandoni, Rothvo{\ss}, and Sanit\`{a} \cite{BGRS}. 
For {\sf ST-MSP} in $\mathbb{R}^2$, Chen and Du \cite{cheng2008relay} applied this method to get the currently 
best known ratio $2.5+\epsilon$. 
In arbitrary metric spaces, the ratio  $\Delta-1$ of \cite{MAN-ZEL} was improved to 
$\lfloor (\Delta+1)/2 \rfloor+1+\epsilon$ in \cite{NY}, also using the same method.
These works assume that {\sf ST-MSP} instances with a constant number of terminals can be solved in polynomial time,
which holds in $\mathbb{R}^2$ if the maximum distance between terminals is polynomial in the number of terminals,
see \cite[Lemma~11]{CHEN-DU} and the discussion there.
In this paper we apply a variant due to Zelikovsky \cite{Zln2}, and obtain the following result.

\begin{theorem} \label{t:main1}
{\sf ST-MSP} with constant $\Delta$ admits an approximation scheme with ratio $1+\ln(\Delta-1)+\epsilon$,
provided that {\sf ST-MSP} instances with a constant number of terminals can be solved in polynomial time.
In particular, in $\mathbb{R}^2$ the ratio is $1+\ln 4+\epsilon < 2.3863 +\epsilon$.
\end{theorem}

We now discuss {\sf SN-MSP} problems with $k=\max_{uv \in V} r_{uv} \geq 2$.
Bredin, Demaine, Hajiaghayi, and Rus \cite{BDHR} considered a related problem of adding a minimum size $S$
such that the unit disc graph of $R \cup S$ is $k$-node-connected
(note that we require $k$-connectivity only between terminals).
For this problem in $\mathbb{R}^2$, they gave an $O(k^5)$-approximation algorithm,
but essentially they implicitly proved that for this class of problems $\alpha=O(\Delta k^3)$.
Recently, it was shown in \cite{NY} that $\alpha=\Theta(\Delta k^2)$ for node-connectivity 
{\sf SN-MSP} instances in any normed space.

Kashyap, Khuller, and Shayman \cite{KKS} considered the $2$-edge/node-connectivity version of {\sf SN-MSP},
where $r_{uv}=2$ for all $u,v \in R$. They used the reduction method described in Definition~\ref{d:K},
namely, their algorithm constructs an {\sf SN} instance as in Definition~\ref{d:K}
and then converts its solution into a bead solution to the {\sf SN-MSP} instance.
Although they analyzed a performance of specific $2$-approximation algorithms -- 
the algorithm of Khuller and Vishkin \cite{KV} for $2$-edge-connectivity and 
the algorithm of Khuller and Raghavachari \cite{KR} for $2$-node-connectivity,
they essentially proved that $\alpha=\Delta$ in both cases.
This implies ratio $2\Delta$ in both cases. 
The analysis of these specific algorithms was recently improved by Calinescu \cite{Calinescu},
showing that their tight performance is $\Delta$ for node-connectivity and 
$2\Delta-1$ for edge-connectivity.
Note that the edge-connectivity version is {\em not} included in our model, 
since in our {\sf SN-MSP} instances every non-terminal node is in $Q$, 
namely, the paths are required to be $S$ disjoint. 

Let $\tau^*=\tau^*(I)$ denote the optimal value of a {\em fractional} bead solution of an {\sf SN-MSP} instance $I$,
namely, $\tau^*$ is the optimum of a standard cut-LP relaxation for the 
corresponding {\sf SN} instance (see Section~\ref{s:main2}).
Here we observe, that if the algorithm we use for the corresponding {\sf SN}
instance computes a solution of cost at most $\rho \tau^*$, then the relevant parameter is the following.

\begin{definition}
For a class ${\cal I}$ of {\sf SN-MSP} instances, let
$\alpha^*=\alpha^*({\cal I})=\sup_{I \in {\cal I}} \frac{{\sf opt}(I)}{\tau^*(I)}$.
\end{definition}


\begin{theorem} \label{t:main2}
For $Q$-connectivity {\sf $\{0,1,2\}$-\sf SN-MSP} $\alpha^*=\frac{\Delta}{2}$.
Thus if $Q$-connecti\-vity $\{0,1,2\}$-{\sf SN} admits a polynomial time algorithm that computes 
a solution of cost at most $\rho \tau^*$, 
then $Q$-connectivity $\{0,1,2\}$-{\sf SN-MSP} admits approximation ratio $\rho \cdot \frac{\Delta}{2}$.
In particular, for $\rho=2$ the ratio is $\Delta$, and thus $\{0,1,2\}$-{\sf SN-MSP} admits a 
$\Delta$-approximation algorithm.
\end{theorem}

Theorems \ref{t:main1} and \ref{t:main2} are proved in 
Sections \ref{s:main1} and \ref{s:main2}, respectively.

\section{Proof of Theorem~\ref{t:main1}} \label{s:main1}

We consider a generic problem defined in \cite{NY}, 
that includes both {\sf ST-MSP} and the classic {\sf Steiner Tree} problem.

\begin{center} 
\fbox{
\begin{minipage}{0.960 \textwidth}
\noindent
{\sf Generalized Steiner Tree}  \\
{\em Instance:} 
\  A (possibly infinite) graph $G=(V,E)$, a finite set $R \subseteq V$ of terminals, and 
a monotone subadditive cost function $c$ on 
subgraphs of $G$. \\
{\em Objective:} Find a minimum-cost connected finite subtree $T$ of $G$ containing~$R$.
\end{minipage}
}
\end{center}

Instead of considering optimal connections only between pairs of terminals, 
we consider optimal connections of terminal subsets of size at most $k$.

\begin{definition} \label{d:Q}
For an instance of {\sf Generalized Steiner Tree} and an integer $k$, $2 \leq k \leq |R|$,
the hypergraph ${\cal H}_k=(R,{\cal E}_k)$ has node set $R$ and hyperedge set 
${\cal E}_k=\{A \subseteq R:2 \leq |A| \leq k\}$.
The cost $c^*(A)$ of $A \in {\cal E}_k$ is the cost of an optimal solution $T_A$ 
to the {\sf Generalized Steiner Tree} instance with terminal set $A$.
\end{definition}

Given a hypergraph ${\cal H}$ with hyperedge costs, 
the {\sf Minimum Connected Spanning Sub-hypergraph} problem seeks a 
minimum cost subset of hyperedges that connects any two nodes.
The construction in Definition~\ref{d:Q} converts the {\sf Generalized Steiner Tree}
problem into the {\sf Minimum Connected Spanning Sub-hypergraph} problem 
in a hypergraph ${\cal H}_k$ of rank $k$.
Any solution of cost $C$ to this problem 
correspond to a solution of value at most $C$ to {\sf Generalized Steiner Tree},
by the aubadditivity and monotonicity of the cost function in the {\sf Generalized Steiner Tree} problem.
The inverse is not true in general, and this reduction 
invokes a fee in the approximation ratio, given in the following definition. 

\begin{definition} \label{d:steiner-ratio}
Given an instance $I$ of {\sf Generalized Steiner Tree} let $\tau_k(I)$ denote the minimum cost
of a connected spanning sub-hypergraph of ${\cal H}_k$.
The {\em $k$-ratio} for a class ${\cal I}$ of {\sf Generalized Steiner Tree} instances is defined by 
$\alpha_k=\sup_{I \in {\cal I}} \frac{\tau_k(I)}{{\sf opt}(I)}$.
\end{definition}

Note that for ${\cal I}$ being the class of {\sf ST-MSP} instances, 
$\alpha_2$ is the parameter $\alpha$ defined in the introduction, and that by \cite{MAN-ZEL} we have
$\alpha_2 = \alpha = \Delta-1$. 
We have $\alpha_k=1$ for instances with $|R|=k$, and in general $\alpha_k$ is monotone decreasing and 
approaching $1$ when $k$ becomes larger. 

In Section~\ref{s:hypergraph} we prove the following statement, which is of independent interest,
and may find applications in other network design problems.

\begin{theorem} \label{t:hypergraph} 
There exists polynomial time algorithm that given 
a hypergraph ${\cal H}=(R,{\cal E})$ with hyper-edge cost $\{c(A):A \in {\cal E}\}$ and 
a spanning tree $T^*$ of (edges of size $2$ of) ${\cal H}$ 
computes a spanning connected sub-hypergraph ${\cal T}$ of ${\cal H}$ of cost at most 
$\tau \left(1+ \ln \frac{c(T^*)}{\tau} \right)$,
where $\tau$ is the minimum-cost of a connected spanning sub-hypergraph of ${\cal H}$.
\end{theorem}

\begin{corollary} \label{c:GST}
For any constant $k$, {\sf Generalized Steiner Tree} admits an approximation ratio
$\alpha_k \left(1+ \ln \alpha_2\right)$, provided that 
for any $A \in {\cal E}_k$, the instance with the terminal set $A$ can be solved in polynomial time.
\end{corollary}
\begin{proof}
By the assumptions, the hypergraph ${\cal H}_k$,
and the costs $c^*(A)$ with the corresponding trees $T_A$ for $A \in {\cal E}_k$, 
can be computed in polynomial time. 
We can also compute in polynomial time an optimal spanning tree $T^*$ in ${\cal H}_2$; 
note that $c(T^*) \leq \alpha_2 \cdot {\sf opt}$.
Then we apply the algorithm in Theorem~\ref{t:hypergraph} to compute 
a sub-hypergraph ${\cal T}$ of ${\cal H}_k$ of $c^*$-cost at most 
$\tau \left(1+ \ln \frac{c(T^*)}{\tau} \right)$,
where $\tau$ is the minimum-cost of a connected spanning sub-hypergraph of ${\cal H}_k$.
Let ${\sf opt}$ denote the optimal solution value for the {\sf Generalized Steiner Tree} instance.
Note that ${\sf opt} \leq \tau \leq \alpha_k {\sf opt}$.
Let $T=\cup_{A \in {\cal T}} T_A$.
Since ${\cal T}$ is a connected hypergraph, $T$ is a feasible solution to the 
{\sf Generalized Steiner Tree} instance.
We have
$c(T) \leq \sum_{A \in {\cal T}} c(T_A)=c^*({\cal T})$,
by the monotonicity and the subadditivity of the $c$-costs. Thus we have:
$$
c(T) \leq c^*({\cal T}) \leq \tau \left(1+ \ln \frac{c(T^*)}{\tau} \right)=
\tau \left(1+ \ln \frac{c(T^*)/{\sf opt}}{\tau/{\sf opt}} \right) \leq
\alpha_k {\sf opt} \left(1+ \ln \alpha_2\right) \ .
$$
\qed
\end{proof}

Du and Zhang \cite{DuZhang} showed that for the classic {\sf Steiner Tree} problem, 
$\alpha_k \leq 1 + 1/ \lfloor \lg k \rfloor$, where $\lg k =\log_2 k$ denotes logarithm base $2$.
In Section~\ref{s:decomposition} 
we prove the following.

\begin{theorem}\label{t:decomposition}
For {\sf ST-MSP}, $\alpha_k \leq 1+\frac{2}{\lfloor\lg \lfloor k/(\Delta-1)\rfloor\rfloor} $ 
for any integer $k \geq 2\Delta-2$.
\end{theorem}

Note that $k \geq \Delta$ is necessary if we want $\alpha_k<2$.
Otherwise, for an instance $I$ of $\Delta$ points on the unit ball 
we have $\frac{\tau(I)}{{\sf opt}(I)}=\frac{k}{\Delta}$,
so $\alpha_k \geq \frac{k}{\Delta}$ if $k \leq \Delta$. 

From Corollary~\ref{c:GST} and Theorem~\ref{t:decomposition} we conclude that for any constant 
$k \geq 2\Delta-2$, it is possible to compute in polynomial time 
a solution to an {\sf ST-MSP} instance of size at most
$\alpha_k  \left(1+ \ln (\Delta-1)\right){\sf opt}$,
where $\alpha_k$ is as in Theorem~\ref{t:decomposition}. 
For the metric space $\mathbb{R}^2$, and given a constant $\epsilon>0$ 
let $k=2^{O(1/\epsilon)}$ with sufficient large constant.
Then by Theorem~\ref{t:decomposition}, 
$\alpha_k \leq 1+\epsilon/\left(1+\ln 4\right)$, and the approximation
ratio of our algorithm is $1+\ln 4 + \epsilon$.
This completes the proof of Theorem~\ref{t:main1}.

\subsection{Proof of Theorem~\ref{t:hypergraph}} \label{s:hypergraph}

For the proof of Theorem~\ref{t:hypergraph} we need the following definition.

\begin{definition}
Given a tree $T=(R,F)$ we say that $A \subseteq R$ {\em overlaps} $F' \subseteq F$ 
if the graph obtained from $T \setminus F'$ by shrinking $A$ into a single node is a tree.
Given edge cost $\{c(e):e \in F\}$ let $F(A)$ be a maximum cost edge set overlapped by $A$. 
\end{definition}

Note that $F \setminus F(A)$ is an edge set of a minimum cost 
spanning tree in the graph obtained from $T$ by shrinking $A$ into a single node;
hence $F(A)$ can be computed in polynomial time.
The following statement appeared in \cite{Z116} (see also \cite{BGRS});
we provide a proof for completeness of exposition.

\begin{lemma}\label{l:sum-bridges}
Let $T=(R,F)$ be a tree with edge costs $\{c(e):e \in F\}$ and let $(R,{\cal E})$ be a connected hypergraph.  
Then $\sum_{A \in {\cal E}} c(F(A)) \geq c(F)$.
Thus there exists $A \in {\cal E}$ such that 
$$\frac{c(F(A))}{c(A)} \geq \frac{c(F)}{c({\cal E})} \ .$$
\end{lemma}
\begin{proof}
For a node $v\in A$, let $C_v$ be the connected component in $T\setminus F(A)$
that contains $v$.
For an edge $e\in F(A)$ that connects two components $C_u,C_v$, let
$y(e)=uv$ be the replacement edge of $e$, of cost $c(y(e))=c(e)$.
The graph $T\cup \{y(e)\}$ contains a single cycle and
$y(e)$ is the heaviest edge in this cycle, since otherwise $F(A)$ is not minimal.
For a hyperedge $A\in {\cal E}$ let $y(A)=\cup_{e\in F(A)}y(e)$ be the replacement 
set of $A$, and let $y({\cal E})=\cup_{A\in {\cal E}} y(A)$.
It is easy to see that $y(A)$ span $A$, and $y({\cal E})$ span $R$.
Consider a MST on $T\cup y({\cal E})$. By the cycle property of a MST,
no edge from $y({\cal E})$ would participate in that MST,
so $c(T)\leq c(y({\cal E}))$. 
Finally, $c(y({\cal E}))=\sum_{A\in {\cal E}}y(A)=\sum_{A\in{\cal E}} c(F(A))$, 
and the lemma follows.
\qed
\end{proof}

\begin{center} 
\fbox{
\begin{minipage}{0.960\textwidth}
\noindent
{\bf Local Replacement Algorithm} \\
{\em Input:} A hypergraph ${\cal H}=(R,{\cal E})$ with hyper-edge cost $\{c(A):A \in {\cal E}\}$, 
and a spanning tree $T^*=(R,F^*)$ of (edges of size $2$ of) ${\cal H}$. \\
{\em Initialization:} ${\cal J} \gets \emptyset$, $F \gets F^*$, $T \gets (R,F)$. \\
{\em While} $c(F)>0$ do: \\
\hphantom{\em While} Find $A \in {\cal E}$ with $\frac{c(F(A))}{c(A)}$ maximum. \\
\hphantom{\em While} - {\em If} \ $c(F(A)) >c(A)$ then do: \\
\hphantom{\em While - {\em If}} - Update $T,{\cal H}$: remove $F(A)$ and shrink $A$ into a single node. \\
\hphantom{\em While - {\em If}} - $F \gets F\setminus F(A)$ and ${\cal J} \gets {\cal J} \cup \{A\}$. \\
\hphantom{\em While} - {\em Else} STOP and {\em Return} ${\cal T}=(R,F \cup {\cal J})$. \\
{\em EndWhile} \\
{\em Return} ${\cal T}=(R,F \cup {\cal J})$.
\end{minipage}
}
\end{center}

At every iteration $|F|$ decreases by at least $1$, hence the algorithm runs in polynomial time,
and clearly it computes a feasible solution. 
We prove the approximation ratio.
Let $F_i$ and ${\cal J}_i$ be the set stored in $F$ and ${\cal J}$, respectively, 
at the beginning of iteration $i+1$,
and let $A_i$ be the hyperedge picked at iteration~$i$.
Denote $f_i=c(F_i)$ and $s_i=c(A_i)$, and recall that $\tau$ denotes 
the minimum cost of a connected spanning sub-hypergraph of ${\cal H}$.
At iteration $i$ we remove $F_{i-1}(A_i)$ from $F_{i-1}$ after verifying that 
$c(F_{i-1}(A_i)) > c(A_i)=s_i$. Hence 
$$f_i \leq f_{i-1}-\max\{c(F_{i-1}(A_i)),c(A_i)\} = f_{i-1}-s_i \cdot \max\left\{\frac{c(F_{i-1}(A_i))}{c(A_i)},1\right\}$$
By Lemma~\ref{l:sum-bridges}, $\frac{c(F_{i-1}(A_i))}{c(A_i)} \geq \frac{f_{i-1}}{\tau}$.
Thus we have
\begin{equation} \label{e:recursion}
f_i \leq f_{i-1}-s_i \cdot \max\{f_{i-1}/\tau,1\} \ .
\end{equation}
The algorithm stops if either $c(F_q)=0$ or $c(F(A))\leq c(A)$ at iteration $q+1$.
In the latter case, $1\geq c(F_q)/\tau$ follows by Lemma~\ref{l:sum-bridges}.
In both cases, we have that there exists an index \ $q$ such that $f_{q-1} > \tau \geq f_q$ holds.
Now we use the following statement from \cite{TAP}.

\begin{lemma} \label{l:TAP}
Let $\tau>0$ and $f_0, \ldots, f_q$ and $s_1, \ldots, s_q$ be sequences of positive 
reals satisfying $f_0 > \tau \geq f_q$, such that (\ref{e:recursion}) holds.
Then $f_q+\sum_{i=1}^q s_i \leq \tau(1+\ln(f_0/\tau))$.
\end{lemma}

Let $q$ be an index such that $f_{q-1} > \tau \geq f_q$ holds.
We may assume that $f_0=c(F^*) > \tau>0$.
Note that  $c({\cal J}_q) = \sum_{i=1}^q s_i$ and that $c(F_i)+c({\cal J}_i) \leq c(F_{i-1})+c({\cal J}_{i-1})$
for any $i$. Hence from Lemma \ref{l:TAP} we conclude that 
$$
c({\cal T}) \leq c(F_q)+c({\cal J}_q) = f_q+\sum_{i=1}^q s_i \leq \tau(1+\ln(f_0/\tau))=
\tau \left(1+ \ln \frac{c(T^*)}{\tau} \right) \ .
$$

This finishes the proof of Theorem~\ref{t:hypergraph}.

\section{Proof of Theorem~\ref{t:main2}} \label{s:main2}

To illustrate our idea, we first prove Theorem~\ref{t:main2} for a particular simple case -- 
the {\sf Steiner Forest with Minimum Number of Steiner Points} ({\sf SF-MSP}) problem, when $r_{uv} \in \{0,1\}$.

\begin{definition}
For a subset $C$ of nodes of a graph $G=(V,E)$ let us use the following notation:
$\Gamma_G(C)$ is the set of neighbors of $C$ in $G$;
$\delta_G(C)=\delta_E(C)$ is the set of edges in $E$ with exactly one endnode in $C$;
$E(C)$ is the set of edges in $E$ with both endnodes in $C$.
Given $R \subseteq V$, 
an {\em $R$-component} of $G$ is a subgraph of $G$ with node set $C \cup \Gamma_G(C)$
and edge set $E(C) \cup \delta_G(C)$, where $C$ is a connected component of $G \setminus R$.
\end{definition}

The cut-LP relaxation for {\sf Steiner Forest} is:
\[ \displaystyle
\begin{array} {llll} 
& \tau^* = & \min          & \ \ \displaystyle \sum_{e \in E} c_e x_e   \\
&          & \ \mbox{s.t.} & \displaystyle \sum_{e \in \delta_E(Y)} x_e \geq f(Y) 
                                     \ \ \ \ \ \ \forall \ \emptyset \neq Y \subset V    \\
&          &               & \ \ \ 0 \leq x_e \leq 1 \ \ \ \ \ \ \ \ \ \ \ \ \ \forall e \in E
\end{array}
\]
where $f(Y)=1$ if there are $u,v \in V$ with $r_{uv}=1$ and $|\{u,v\} \cap Y|=1$,
and $f(Y)=1$ otherwise. 

Robins and Salowe \cite{ROB-SAL} proved that if $V$ is a set of ponts in a metric space,
then there exists a tree $T=(V,E)$ of minimum total length $\sum_{uv \in E} d(u,v)$
that has maximum degree $\leq \Delta$.
Since any inclusion-minimal solution to a {\sf Steiner Forest} instance is a forest,
this implies the following. 

\begin{lemma} \label{l:delta}
For any instance of {\sf SF-MSP} there exists an optimal solution $S,G$
such that $G$ has maximum degree $\Delta$. \qed
\end{lemma}

The following statement was first observed in \cite{KKS}.

\begin{lemma} \label{l:KKS} 
Let $R$ be a set of terminals and $S$ a set of points in a normed space such that 
the unit-disc graph of $R \cup S$ contains a tree $T$ with leaf set $R$.
Let $S'$ be obtained from $S$ by replacing each $v \in S$ by $\deg_T(v)$ copies of $v$.
Then the unit disc graph of $R \cup S'$ contains a simple cycle on $R \cup S'$.
\end{lemma}
\begin{proof}
Traverse the tree $T$ in a DFS order; each time a node $v \in S$ is visited, choose a different copy of $v$.
\qed
\end{proof}

Given a tree $T$, we will call a cycle as in the lemma above a {\em DFS cycle} of $T$.

Now we can prove Theorem~\ref{t:main2} for the {\sf SF-MSP} case.
Let $S$ be an inclusion minimal solution to an {\sf SF-MSP} instance.
By Lemma~\ref{l:delta}, the unit-disc graph of $R \cup S$ contains an $r$-connected forest 
$H$ such that $\deg_H(v) \leq \Delta$ for every $v \in S$. 
Every $R$-component $T$ of $H$ (a.k.a. full Steiner component) is a tree with leaf set 
in $R$ and all internal nodes in $S$. It is easy to see that by replacing every 
$R$-component $T$ of $H$ by a DFS cycle of capacity $1/2$ results in a feasible solution to
the cut-LP relaxation, which proves Theorem~\ref{t:main2} for the {\sf SF-MSP} case.

Now we prove Theorem~\ref{t:main2} for {\sf $\{0,1,2\}$-\sf SN-MSP}.
We start by describing the cut-LP relaxation for {\sf SN}.
We need some definitions.

\begin{definition}
An ordered pair $\hat{X}=(X,X^+)$ of subsets of a groundset $V$ is called a {\em biset} if $X \subseteq X^+$; 
$X$ is the {\em inner part} and $X^+$ is the {\em outer part} of $\hat{X}$,
$\Gamma(\hat{X})=X^+ \setminus X$ is the {\em boundary} of $\hat{X}$, 
and $X^*=V \setminus X^*$ is the complementary set of $\hat{X}$. 
An edge $e=uv$ covers a biset $\hat{X}$ if it has one endnode in $X$ and the other in $V \setminus X^+$.
For a biset $\hat{X}$ and an edge-set/graph $J$ let $\delta_J(\hat{X})$ denote the set of edges in 
$J$ covering $\hat{X}$.
\end{definition}

By Menger's Theorem, a graph $G=(V,E)$ is $(r,Q)$-connected if, and only if, 
$|\delta_E(\hat{Y})| \geq f(\hat{Y})$, where $f$ is a biset-function defined by 
$$
f(\hat{Y}) = \left \{ \begin{array}{ll}
\max\limits_{uv \in \delta_{D_r}(\hat{Y})} r_{uv} - |\Gamma(\hat{Y})|  \ \ \  & \mbox{if} \ \ \Gamma(\hat{Y}) \subseteq Q \\
\ \ \ 0                                                                      & \mbox{ otherwise}
\end{array} \right .
$$
The cut-LP relaxation for {\sf SN} is

\[ \displaystyle
\begin{array} {llll} 
& \tau^* = & \min          & \ \ \displaystyle \sum_{e \in E} c_e x_e   \\
&          & \ \mbox{s.t.} & \displaystyle \sum_{e \in \delta_E(\hat{Y})} x_e \geq f(\hat{Y}) 
                                     \ \ \ \ \ \ \forall \mbox{ biset } \hat{Y}    \\
&          &               & \ \ \ 0 \leq x_e \leq 1 \ \ \ \ \ \ \ \ \ \ \ \ \ \forall e \in E
\end{array}
\]

We will say that a graph with edge capacities $x_e$ is {\em fractionally $(r,Q)$-connected} if $x$
is a feasible solution to the above cut-LP relaxation.

To prove Theorem~\ref{t:main2}, we prove in the next sections the following 
two theorems about $\{0,1,2\}$-connected graphs,
that are of independent interest, and may find further applications in low connectivity network design. 
An $r$-connected graph $G$ is {\em minimally $r$-connected} if no proper subgraph of $G$ is $r$-connected.

\begin{theorem} \label{t:C}
Let $G$ be a minimally $(r,Q)$-connected graph such that $Q \cup R=V$ and 
$r_{uv} \in \{0,1,2\}$ for all $u,v \in R$.
Then every $R$-component is a tree.
Furthermore, for any subset ${\cal C}$ of connected components of $G \setminus R$,
replacing for each $C \in {\cal C}$ the corresponding tree by a DFS cycle of capacity $1/2$ 
results in a fractionally $(r,Q)$-connected graph.
\end{theorem}

\begin{theorem} \label{t:D}
Let $R$ be a set of terminals in a normed space, let $B \subseteq R$,
and let $r$ be a $\{0,1,2\}$ requirement function on $R$.
Let $S$ be an inclusion minimal set of points  
such that the unit-disc graph of $R \cup S$ is $(r,B \cup S)$-connected.
Among all $(r,B \cup S)$-connected spanning subgraphs of the unit-disc graph of $R \cup S$, 
let $G=(V,E)$ be one of minimum total length $\sum_{uv \in E} d(u,v)$.
Then $\deg_G(v) \leq \Delta$ for all $v \in S$. 
\end{theorem}

Particular cases of Theorem~\ref{t:D} were proved by Robins and Salowe \cite{ROB-SAL} for $r \equiv 1$,
and by Calinescu \cite{Calinescu} for $r \equiv 2$.
We prove Theorems \ref{t:C} and \ref{t:D} in Sections \ref{s:C} and \ref{s:D}, respectively,
relying on these particular cases.  
From Theorem~\ref{t:C}, Theorem~\ref{t:D}, and Lemma~\ref{l:KKS}, we obtain the following corollary,
that implies Theorem~\ref{t:main2}.

\begin{corollary}
For any feasible solution $S,G$ to an instance of $\{0,1,2\}$-{\sf SN-MSP}
there exists a half integral bead solution of value at most $\Delta|S|/2$.
\end{corollary}


\subsection{Proof of Theorem~\ref{t:C}} \label{s:C}

A {\em block} of a graph $G$ is an inclusion-maximal $2$-connected subgraph of $G$, 
or a graph induced by a bridge of $G$. 
It is known that every edge belongs to exactly one block, hence 
the blocks of a graph partition its edge set. Furthermore, any two blocks have at most one node in common.

\begin{lemma} \label{l:block}
Let $G=(V,E)$ be a minimally $(r,Q)$-connected graph such that
$r_{uv} \in \{0,1,2\}$ for all $uv \in D_r$ and $Q \cup R = V$. 
Let $G'=(V',E')$ be a $2$-connected block of $G$ and let $R'=R \cap V'$. 
Then $|R \cap V'| \geq 2$ and 
no proper $2$-connected subgraph of $G'$ that contains $R'$ exists.
\end{lemma}
\begin{proof}
We may assume that $G$ is connected, as otherwise we may consider each connected component of $G$ separately.
Any $V'$-component $C$ has exactly one node in $V'$, which we call the {\em attachment node} of $C$. 
Note that if $r_{uv}=2$ such that $v$ belongs to a $V'$-components $C_v$ of $G$ and $u \notin C_v$,
then the attachment node of $C_v$ is in $V \setminus Q$, and hence is in $R$, by the assumption $Q \cup R=V$.

We prove that $|R'| \geq 2$ 
Since $G'$ is $2$-connected, and $G$ is minimally $(r,Q)$-connected,
there exists $uv \in D_r$ with $r_{uv}=2$ such that $u \in V'$, 
or $u,v$ belong to disjoint $V'$-components.
Suppose that $u \in V'$. If $v \in V'$ then we are done.
Else, $v$ belongs to a $V'$-component, and the attachment node of this component is in $R$.
If $u,v$ belong to disjoint $V'$-components, then the attachment nodes of these components are 
distinct and belong to $R$. In all cases, we have $|R'| \geq 2$.

We prove that if $G''=(V'',E'')$ is a $2$-connected subgraph of $G'$ that contains $R'$, then $G''=G'$.
Suppose that $G'' \neq G'$.
Let $A$ be the set of attachment nodes that are in $V' \setminus V''$. 
Note that $A \subseteq Q \setminus R$.
In $G'$, shrink $V''$ into a single node $v''$, 
and take $F$ to be the edge set of some inclusion minimal tree in $G'$ 
that contains $A \cup \{v''\}$. Let $I=E'' \cup F$.
If $A=\emptyset$ then $F=\emptyset$, and $I=E''$. 
Otherwise, there is $a \in A$ that has degree exactly $1$ in $(V',I)$.
In both cases, $I$ must be a proper subset of $E' \setminus E''$. 
Let $\hat{G}$ be obtained from $G$ by replacing $E'$ by $I$. 
It is not hard to verify that $\hat{G}$ is $(r,Q)$-connected, since $A \subseteq Q \setminus R$.
Furthermore, $\hat{G}$ is a proper subgraph of $G$, since $I$ is a proper subset of $E' \setminus E''$
This contradicts the minimality of $G$.
\qed
\end{proof}

A path $P$ is an {\em $L$-chord path} of a cycle $L$ in a graph $G$ if the endnodes of $P$ are in $L$
but no internal node of $P$ is in $L$. 
Relying on ear decomposition of $2$-connected graphs, Calinescu \cite{Calinescu} proved the following.

\begin{lemma} [\cite{Calinescu}] \label{l:calinescu-cycle}
Let $G'=(V',E')$ be a $2$-connected graph and let $R' \subseteq V$ with $|R'| \geq 2$.
Suppose that no proper $2$-connected subgraph of $G$ that contains $R'$ exists.
Then any cycle $L$ in $G'$ contains at least $2$ nodes in $R'$, and any $L$-chord path
contains at least one node in $R'$ that does not belong to $L$.\footnote{
This statement is not true for edge-connectivity; for example, if $R=\{s,t\}$ and 
$G$ consists of $2$ edge-disjoint $st$-paths that have $2$ nodes $u,v$ in common, 
then the simple cycle that contains $u,v$ contains no node from $R$.}
\end{lemma}

We generalize this to $\{0,1,2\}$-$Q$-connectivity, as follows.

\begin{lemma} \label{l:cycle}
Let $G=(V,E)$ be a minimally $(r,Q)$-connected graph such that
$r_{uv} \in \{0,1,2\}$ for all $uv \in D_r$ and $Q \cup R = V$. 
Then any cycle $L$ in $G$ contains at least $2$ nodes in $R$, and any $L$-chord path
contains at least one node in $R$ that does not belong to $L$.
\end{lemma}
\begin{proof}
Let $L$ be a cycle in $G$. Then $L$ is contained in some $2$-connected block $G'=(V',E')$ of $G$;
moreover, any $L$-chord path is also contained in $G'$. Let $R'=R \cap V'$. 
By Lemma~\ref{l:block}, $G',R'$ satisfy the conditions of Lemma~\ref{l:calinescu-cycle}; 
hence the statement follows from Lemma~\ref{l:calinescu-cycle}.
\qed
\end{proof}

By Lemma~\ref{l:cycle}, the graph $G \setminus R$ is a forest, 
and every $v \in R$ has at most one neighbor in each connected component of $G \setminus R$. 
This implies the first part of Theorem~\ref{t:C}.
Now we prove the second part, namely, the following. 

\begin{lemma} \label{l:replace-cycle}
Let $G$ be a minimally $(r,Q)$-connected graph such that $r_{uv} \in \{0,1,2\}$ for all $uv \in D_r$
and $Q \cup R = V$.
Then for any subset ${\cal C}$ of connected components of $G \setminus C$,
replacing for each $C \in {\cal C}$ the corresponding tree $T_C$ by a DFS cycle on $\Gamma_G(C)$ of capacity $1/2$ 
results in a fractionally $(r,Q)$-connected graph $H$.
\end{lemma}
\begin{proof}
Suppose to the contrary that there exists $uv \in D_r$ such that $u,v$ are not 
fractionally $(r_{uv},Q)$-connected in $H$.
This may happen only if $r_{uv}=2$ and there exists $C \in {\cal C}$ such that $u,v$ can be disconnected 
by removing two elements $a,b$ of $T_C$ from $G$. If one of $a,b$ is an edge
we can replace it by its endnode in $T_C$, hence we may assume that each of $a,b$ is a node.
Note that $a \neq b$, since otherwise $u,v$ can be disconnected by removing 
the single element $a$, contradicting that $\lambda_G^Q(u,v) \geq r_{uv}=2$.
Let $P_{ab}$ be the $ab$-path in $T_C$. Note that all the internal nodes of $P_{ab}$ are in $C$,
so none of them is a terminal. 
Consider two $(Q \cup E)$-disjoint $u,v$ paths in $G$.
One of them must contain $a$ and the other contains $b$; denote these paths by $P_a$ and $P_b$, respectively.
The union of the paths $P_a$ and $P_b$ contains a simple cycle $L$ that contains $a,b$.
Hence the path $P_{ab}$ has a subpath $P$ such that $P$ is an $L$-chord path.
This contradicts Lemma~\ref{l:cycle}, since no internal node of $P$ is a terminal. 
\qed
\end{proof}

The proof of Theorem~\ref{t:C} is complete.


\section{Proof of Theorem~\ref{t:decomposition}} \label{s:decomposition}

For a tree $T=(V,F)$ and $A \subseteq V$ let $T_A=(V_A,F_A)$ be 
the inclusion minimal subtree of $T$ that contains $A$.
To prove Theorem~\ref{t:decomposition} it is sufficient to prove the following.

\begin{lemma} \label{l:decomposition}
Let $T=(V,F)$ be a tree of maximum degree $\Delta \geq 2$, let $R \subseteq V$, and let $S=V \setminus R$. 
Then for any integer $k \geq 2\Delta-2$ there exists a connected hypergraph 
${\cal H}=(R,{\cal E})$ of rank $\leq k$ such that 
$\sum_{A \in {\cal E}} |V_A \cap S| \leq \left(1+\frac{2}{\lfloor\lg \lfloor k/(\Delta-1)\rfloor \rfloor}\right)|S|$.
\end{lemma}

To prove Lemma~\ref{l:decomposition} we prove the following.

\begin{lemma}\label{l:e-decomposition}
Let $T=(V,F)$ be a tree with edge costs $\{c(e) \geq 1:e \in F\}$ and let $R \subseteq V$.
Then for any integer $p \geq 2$ there exists a connected hypergraph 
${\cal H}=(R,{\cal E})$ of rank $\leq p$ such that 
$\sum_{A \in {\cal E}}c(F_A)+|{\cal E}|-1 \leq \left(1+\frac{2}{\lfloor\lg p\rfloor}\right)c(T)$.
\end{lemma}

Lemma~\ref{l:e-decomposition} will be proved later. 
Now we show that it implies Lemma~\ref{l:decomposition}.
An {\em $R$-component} of $T$ is a maximal inclusion subtree
of $T$ such that all its leaves are in $R$ but no its internal node is in $R$.
It is easy to see that it is sufficient to prove Lemma~\ref{l:decomposition} for each $R$-component separately,
hence we may assume that $R$ is the set of leaves of $T$.

If $T$ is a star, then since $k \geq 2\Delta-2 \geq \Delta$, 
we let ${\cal E}$ to consist of a single hyperedge $A=R$.
Then $|V_A \cap S|=1=|S|$, and Lemma~\ref{l:decomposition} holds in this case.

Henceforth assume that $T$ is not a star.
For $v \in S$ let $R(v)$ be the set of neighbors of $v$ in $R$, 
and note that $|R(v)| \leq \Delta-1$.
Let $T'=(V',F')=T \setminus R$ and let $R'=\{v \in S: R(v) \neq \emptyset\}$. 
Applying Lemma~\ref{l:e-decomposition} on $T'$ with unit edge-costs and $R'$, we obtain that 
for $p=\lfloor k/(\Delta-1) \rfloor$ there exists a connected hypergraph 
${\cal H}'=(R',{\cal E}')$ of rank $\leq p$ such that 
$\sum_{A' \in {\cal E}'}|F'_{A'}|+|{\cal E}'|-1 \leq \left(1+\frac{2}{\lfloor\lg p\rfloor}\right)|F'|$.
Note that $|F'|=|V'|-1$ and that $|V'_{A'}|=|F'_{A'}|-1$ for every $A' \in {\cal E}'$. 
Hence 
$$\sum_{A' \in {\cal E}'}|V'_{A'}|-1 \leq \left(1+\frac{2}{\lfloor\lg p\rfloor}\right)(|V'|-1) \leq
  \left(1+\frac{2}{\lfloor\lg p\rfloor}\right)|V'|-1 \ .$$
For $A' \in {\cal E}'$ let $A=\cup_{v \in A'}R(v)$; then $|A| \leq p(\Delta-1)$. 
Let ${\cal E}=\{A:A' \in {\cal E}'\}$.
Then ${\cal H}=(R,{\cal E})$ is a connected hypergraph of rank $\leq p(\Delta-1) \leq k$, and
$$\sum_{A \in {\cal E}} |V_A \cap S| = \sum_{A' \in {\cal E}'} |V'_{A'}|
\leq \left(1+\frac{2}{\lfloor\lg p\rfloor}\right)|V'| = 
\left(1+\frac{2}{\lfloor\lg \lfloor k/(\Delta-1)\rfloor \rfloor}\right)|S| \ .$$

In the rest of this section we prove Lemma~\ref{l:e-decomposition}, by 
extending the proof of Du and Zhang \cite{DuZhang} of an existence of a connected 
hypergraph ${\cal H}=(R,{\cal E})$ of rank $\leq p$ such that 
$\sum_{A \in {\cal E}}c(F_A) \leq \left(1+\frac{1}{\lfloor\lg p\rfloor}\right)c(T)$.
We have an extra term of $|{\cal E}|-1$, and we show that 
this term can be bounded by $\frac{c(T)}{\lfloor\lg p\rfloor}$.

We start by transforming the tree into a (rooted) binary tree $T$ with edge-costs, 
which node set is partitioned into a set $R$ of terminals and a set $S$ of non-terminals, 
such that the following properties hold:
\begin{itemize}
\item[(A)]
$R$ is the set of leaves of $T$.
\item[(B)]
The cost of any edge of $T$ is either $0$ or is at least $1$,
and among the edges that connect a node in $S=V \setminus R$ to its children, 
at most one has cost $0$.
\item[(C)]
$T$ is a full binary tree, namely, every $v \in S$ has exactly $2$ children.
\end{itemize}
To obtain such a tree, root $T$ at an arbitrary non-leaf node $\hat{s} \in S=V\setminus R$,
and apply the following standard reductions. 
\begin{enumerate}
\item
While $T$ has a leaf in $S$, remove this leaf; hence every leaf of $T$ is in $R$.  
Then, for every $v \in R$ that is not a leaf, add to $T$ a new node $v'$ and an edge $vv'$ of cost $0$,
add $v'$ to $R$, and move $v$ from $R$ to $S$. 
After this step, properties (A) and (B) hold.
\item
While there is $v \in S$ that has one child, replace the path $P$ of length $2$ that contains $v$
by a single edge of cost $c(P)$, and exclude $v$ from $S$.
After this step, every $v \in S$ has at least $2$ children.
\item
While there is $v \in S$ that has more than $2$ children, do the following.
Let $u$ be a child of $v$ such that the cost of the edge $vu$ is at least $1$. 
Add a new node $v'$ and the edge $vv'$ of cost $0$,
and for every child of $u'$ of $v$ distinct from $u$ replace the edge $vu'$ by the edge $vu'$.
After this step, all the three properties (A), (B), and (C) hold.
\end{enumerate}

Consequently, to prove Lemma~\ref{l:e-decomposition},
it is sufficient to prove the following.

\begin{lemma}\label{l:abc-decomposition}
Let $T=(V,F)$ be a tree with edge costs $\{c(e):e \in F\}$ and leaf set $R$, satisfying properties (A),(B),(C),
Then for any integer $p \geq 2$ there exists a connected hypergraph 
${\cal H}=(R,{\cal E})$ of rank $\leq p$ such that 
$\sum_{A \in {\cal E}}c(F_A)+|{\cal E}|-1 \leq \left(1+\frac{2}{\lfloor\lg p\rfloor}\right)c(T)$.
\end{lemma}

Let $T=(V,F)$ be a rooted tree with leaf set $R$ and let $S=V \setminus R$. 
For two nodes $u,v$ of $T$ let $P_T(u,v)$ denote the unique path in $T$ between $u$ and $v$.

\begin{definition}
We say that $T$ is {\em proper} if every node in $S$ has at least $2$ children.
We say that a mapping $f:S \rightarrow R$ is $T$-proper if
\begin{itemize}
\item 
For every $u \in S$, $f(u)$ is a descendant of $u$.
\item 
The paths $\{P_T(u,f(u)):u \in S\}$ are edge disjoint.
\end{itemize}
Given a subtree $T'$ of $T$ with leaf set $L'$ and a proper mapping $f$,
the set of {\em terminal connecting paths} of $T'$ is  $\{P_T(u,f(u)): u \in L' \setminus R\}$.
Let $\hat{T}'$ denote the tree obtained from $T'$ by adding to $T'$ all the terminal connecting paths. 
\end{definition}

Du and Zhang \cite{DuZhang} proved that any proper tree $T$ admits a proper mapping.
We prove the following.

\begin{lemma} \label{l:mapping}
Let $T=(V,F)$ be a proper tree and let $F_1 \subseteq F$ be such that 
any $u \in S$ has a child connected to $u$ by an edge in $F_1$.
Then there exists a $T$-proper mapping $f$ such that for every $u \in S$,
the path $P_T(u,f(u))$ contains at least one edge in $F_1$.
\end{lemma}
\begin{proof}
The proof is by induction on the height of the tree.
Let $T$ be a tree as in the lemma of height $h$. 
If $h=1$, then $T$ has one internal node (the root), say $u$, and we set $f(u)$ 
to be the node that is connected to $u$ by an edge in $F_1$.
Suppose that the statement is true for trees with height $h-1 \geq 1$, 
and we prove it for trees of height $h$.
Let $T'$ be obtained from $T$ by removing nodes of distance $h$ from the root.
By the induction hypothesis, for $T'$ there exists a mapping $f'$ as in the lemma.
Let $u$ be an internal node of $T$. Consider two cases.

Suppose that $u$ is an internal node of $T'$. 
If $f'(u)$ is a leaf of $T$, then define $f(u)=f'(u)$. 
If $f'(u)$ is an internal of $T$, then  $f'(u)$ is a leaf of $T'$,
and all its children in $T$ are leaves.
Then we set $f(u)$ to be a child of $f'(u)$
that is connected to $f'(u)$ by an edge in $F_1$

Suppose that $u$ is a leaf of $T'$. 
Then the children of $u$ in $T$ are leaves, and we set $f(u)$ to be a child of $u$
that is connected to $u$ by an edge in $F_1$.

It is easy to verify that the obtained mapping $f$ meets the requirements.
\qed
\end{proof}

The following statement is implicitly proved by Du and Zhang \cite{DuZhang}.

\begin{lemma} [\cite{DuZhang}] \label{l:subtrees}
Let $T$ be a proper binary tree with non-negative edge costs and let $f$ be a proper mapping.
Then for any integer $p \geq 2$ 
there exists an edge-disjoint partition ${\cal T}$ of $T$ into subtrees 
such that the following holds:
\begin{itemize} 
\item[{\em (i)}]
The hypergraph with node set $R$ and hyperedge set 
${\cal E}=\{\hat{T}' \cap R:T' \in {\cal T}\}$ is connected and has rank at most $p$.
\item[{\em (ii)}]
The total number of terminal connecting paths of all subtrees in ${\cal T}$ is at least $|{\cal T}|-1$,
and their total cost is at most $c(T)/\lfloor \lg p \rfloor$.
\end{itemize}
\end{lemma}

We now finish the proof of Lemma~\ref{l:abc-decomposition}, and thus also of Lemma~\ref{l:e-decomposition}.
Let $F_1=\{e \in F:c(e) \geq 1\}$ and let $f$ be a proper mapping as in Lemma~\ref{l:mapping}.
Let ${\cal T}$ be a partition as in Lemma~\ref{l:subtrees}, 
and let ${\cal E}$ be as in Lemma~\ref{l:subtrees}(i), so the hypergraph ${\cal H}=(R,{\cal E})$
is connected and has rank at most $p$.
By Lemma~\ref{l:subtrees}(ii), 
the total number of terminal connecting paths of all subtrees is at least $|{\cal T}|-1=|{\cal E}|-1$,
while their total cost is at most $c(T)/\lfloor \lg p \rfloor$.
Every terminal connecting path contains an edge from $F_1$, by Lemma~\ref{l:mapping}, 
and thus has cost at least $1$.
Hence the total cost of all terminal connecting paths is at least $|{\cal E}|-1$.
Consequently 
$$|{\cal E}|-1 \leq \frac{c(T)}{\lfloor\lg p\rfloor} \ .$$

For $A=\hat{T}' \cap R \in {\cal E}$
let $P(T')$ denote the union of the edge sets of the terminal connecting paths of $T'$.
Then $c(F_A) \leq c(\hat{T}') = c(T)+c(P(T'))$, hence
$$
\sum_{A \in {\cal E}} c(F_A) \leq \sum_{T' \in {\cal T}} [c(T')+c(P(T'))]=
\sum_{T' \in {\cal T}} c(T') + \sum_{T' \in {\cal T}} c(P(T')) \leq c(T)+\frac{c(T)}{\lfloor\lg p\rfloor} \ .
$$

Summarizing, we have 
$$\sum_{A \in {\cal E}}c(F_A)+|{\cal E}|-1 \leq 
c(T)+\frac{c(T)}{\lfloor\lg p\rfloor}+\frac{c(T)}{\lfloor\lg p\rfloor} =
\left(1+\frac{2}{\lfloor\lg p\rfloor}\right)c(T) \ .$$

The proof of Lemma~\ref{l:abc-decomposition}, and thus also of Lemma~\ref{l:e-decomposition}
and Theorem~\ref{t:decomposition} is now complete.

\section{Proof of Theorem~\ref{t:D}} \label{s:D}

To prove Theorem~\ref{t:D}, we use the following result of Calinescu \cite{Calinescu}.

\begin{lemma}[\cite{Calinescu}] \label{l:D-Calinescu}
Let $R'$ be a set of terminals in a normed space and
let $S'$ be an inclusion minimal set of points  
such that the unit-disc graph of $R' \cup S'$ is $2$-connected.
Among all $2$-connected spanning subgraphs of the unit-disc graph of $R' \cup S'$, 
let $G'=(V',E')$ be one of minimum total length $\sum_{uv \in E'} d(u,v)$.
Then $\deg_{G'}(v) \leq \Delta$ for all $v \in S'$. 
\end{lemma}

Let $G,B,S,r$ be as in Theorem~\ref{t:D}.
As in the proof of Theorem~\ref{t:C}, we may assume that $G$ is connected.
Consider a $2$-connected block $G'=(V,E')$ of $G$.
Let $R'=R \cap V'$ and $S'=S \cap V'$. 
Then by Lemma~\ref{l:block} no proper $2$-connected subgraph of $G'$ that contains $R'$ exists,
hence $S'$ is an inclusion minimal set of points  
such that the unit-disc graph of $R \cup S$ is $2$-connected.
Furthermore, since $G$ has minimum total length, so is $G'$.
Thus by Lemma~\ref{l:D-Calinescu}, $\deg_{G'}(v) \leq \Delta$ for all $v \in S'$.
Consequently, $\deg_G(v) \leq \Delta$ holds for any $s \in S$ that belongs to exactly one block of $G$.
A node $s$ is a {\em cut-node} of a connected graph if its removal disconnects the graph.
It is known that $s$ is a cut-node of a graph if and only if $s$ belongs to at least two blocks of the graph.
Our goal now is to show that $\deg_G(v) \leq \Delta$ holds for any cut-node $s \in S$ of $G$.

Let $s \in S$ be a cut-node of $G$. Suppose to the contrary that $\deg_G(v) \geq \Delta+1$. 
Then by \cite{ROB-SAL} there are neighbors $a,b$ of $s$ in $G$ such that $d(a,b) \leq d(a,s)$.
By a reduction from \cite{ROB-SAL,Calinescu}, we may assume that all the lengths of the edges in $G$ are distinct,
hence $d(a,b) < d(a,s)$.
Let $H$ be obtained from $G$ by replacing the edge $sa$ by the edge $ab$.
We claim that $H$ is $(r,Q)$-connected, which gives a contradiction, 
since $H$ has smaller total length than $G$. 
Thus to finish the proof of Theorem~\ref{t:D}, it is sufficient to prove the following.

\begin{lemma} \label{l:replace}
Let $G=(V,E)$ be an $(r,Q)$-connected graph with $r_{uv} \in \{0,1,2\}$ for all $uv \in D_r$,
and let $sa,sb \in E$ be a pair of $(r,Q)$-connectivity critical edges with $s \in Q \setminus R$.
Then the graph $H$ obtained from $G$ by replacing the edge $sa$ by the edge $ab$ is also $(r,Q)$-connected.
\end{lemma}

\begin{figure} \label{f:cases2}
\centering
\epsfbox{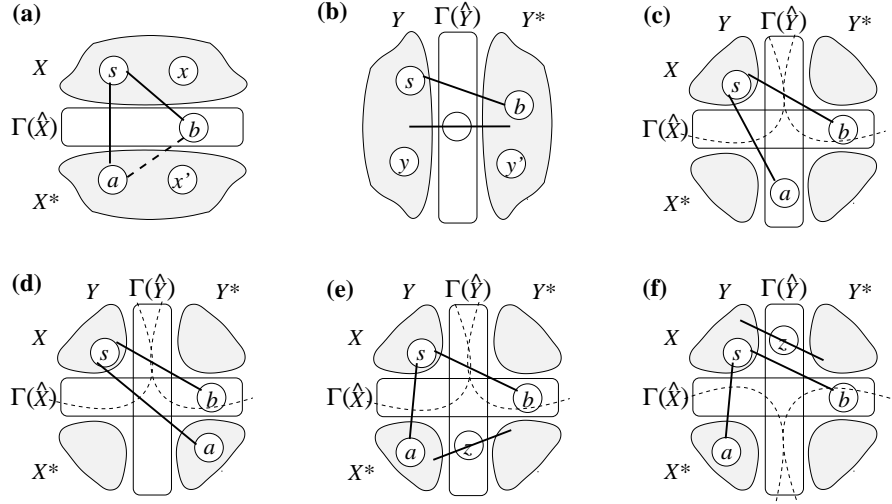}
\caption{Illustration to the proof of Lemma~\ref{l:replace}.}
\end{figure}

\begin{proof}
Suppose to the contrary that there is $xx' \in D_r$ such that $\lambda_{H}^Q(xx') \leq r_{xx'}-1$.
It is easy to see that any $u,v$ that are connected in $G$ also connected in $H$, hence we must have 
$r_{xx'}=2$. Consider the graph $J=G \setminus \{sa\}$.
Since $\lambda_{J \cup \{sb\}}^Q(x,x')=1$ and $\lambda_{J \cup \{sa\}}^Q(x,x')=2$, 
then by Menger's Theorem, there exists a biset $\hat{X}$ such that 
$s \in X$, $a \in X^*$, $b \in \Gamma(\hat{X}) \subseteq Q$, $\delta_G(\hat{X})=\{sa\}$, 
and one of $x,x'$ belongs to $X$ and the other to $X^*$, 
say $x \in X$ and $x' \in X^*$; see Figure~\ref{f:cases2}(a).
Similarly, since the edge $sb$ is $(r,Q)$-connectivity critical, there exist $yy' \in D_r$ with 
$r_{yy'}=2$ and a biset $\hat{Y}$, such that $s \in Y$, $b \in Y^*$,  
$sb \in \delta_G(\hat{Y})$, $|\delta_G(\hat{Y})|+|\Gamma(\hat{Y})|=2$, $\Gamma(\hat{Y}) \subseteq Q$, 
and one of $y,y'$ belongs to $Y$ and the other to $Y^*$, 
say $y \in Y$ and $y' \in Y^*$; see Figure~\ref{f:cases2}(b).
Now we consider the three cases, $a \in \Gamma(\hat{Y})$, $a \in Y^*$, and $a \in Y$, 
and at each of them arrive to a contradiction.

Suppose that $a \in \Gamma(\hat{Y})$; see Figure~\ref{f:cases2}(c).
Then $x \notin \Gamma(\hat{Y})$, so $x \in X \cap Y$ or $x \in X \cap Y^*$.
If $x \in X \cap Y^*$ then the biset $\hat{Z}=\hat{X} \setminus \hat{Y}$ satisfies
$|\Gamma(\hat{Z})|+|\delta_G(\hat{Z})|=1$
(since $\Gamma_G(\hat{Z})=\{b\}$ and $\delta_G(\hat{Z})=\emptyset$),
$x \in Z$, and $x' \in Z^*$;
this contradicts the assumption $\lambda_G^Q(x,x')=2$.
In the case $x \in X \cap Y$, we obtain a similar contradiction for 
the biset $\hat{Z}=(X \cap Y \setminus \{s\},X \cap Y)$.

The analysis of the case  $a \in Y^*$, see Figure~\ref{f:cases2}(d), 
is similar to that of the case $a \in \Gamma(\hat{Y})$.

Now suppose that $a \in Y$; see Figure~\ref{f:cases2}(e,f).
Since $|\Gamma(\hat{Y})|+|\delta_G(\hat{Y})|=2$ and since $sb \in \delta_G(\hat{Y})$, 
there is another element $z \in \Gamma(\hat{Y}) \cup \delta_G(\hat{Y})$.
Note that if $z$ is a node then $z \in X^* \cap \Gamma(\hat{Y})$ (Figure~\ref{f:cases2}(e))
or $z \in X \cap \Gamma(\hat{Y})$ (Figure~\ref{f:cases2}(f)).
If $z$ is an edge then $z$ connects $Y \cap X^*$ and $Y^* \setminus X$ (Figure~\ref{f:cases2}(e))
or $X \cap Y$ and $Y^* \setminus X^*$ (Figure~\ref{f:cases2}(f)).
In the cases in Figure~\ref{f:cases2}(e), when $z \in X^* \cap \Gamma(\hat{Y})$ is a node, 
or $z$ is an edge that connects $Y \cap X^*$ and $Y^* \setminus X$, the contradiction is obtained 
in the same way as in the case $a \in \Gamma(\hat{Y})$.
We therefore are left with the cases in Figure~\ref{f:cases2}(f), when $z \in X \cap \Gamma(\hat{Y})$
or $z$ is an edge that connects $X \cap Y$ and $Y^* \setminus X^*$.
Then we consider the location of $x'$. Note that $x' \notin \Gamma(\hat{Y})$, 
hence $x' \in Y$ or $x' \in Y^*$. 
In the case $x' \in Y$ we obtain a contradiction by considering the biset $\hat{Z}=\hat{Y} \setminus \hat{X}$, and
in the case $x' \in Y^*$ we obtain a contradiction by considering the biset $\hat{Z}=\hat{X} \cup \hat{Y}$.
\qed
\end{proof}

The proof of Theorem~\ref{t:D} is complete.

\section{Conclusions}
In this paper we considered the {\sf Survivable Network with Minimum
Number of Steiner Points} problem in a normed space.
The main results of this paper are 
a $(1+\ln (\Delta-1) + \epsilon)$-approximation scheme for {\sf ST-MSP}, and 
a $\Delta$-approximation algorithm for {\sf $\{0,1,2\}$-SN-MSP}.
For {\sf ST-MSP} in $\mathbb{R}^2$ this improves the ratio $2.5+\epsilon$ of \cite{cheng2008relay}.
For {\sf $\{0,1,2\}$-SN-MSP}, no nontrivial approximation algorithm was known before,
but for the specific case of {\sf SF-MSP} this improves the ratio $2\Delta$ that can be deduced
from the work of \cite{KKS}. Obtaining even better approximation ratios is an important future work.


\end{document}